\documentclass[preprint,12pt]{elsarticle}

\usepackage{amssymb,amsmath}
\usepackage{multirow}
\usepackage{algorithm}
\usepackage{algorithmic}
\usepackage{color}
\usepackage{graphicx}
\usepackage{hyperref}
\usepackage{titlesec}
\titleformat{\section}{\bfseries\large}{\thesection.}{1em}{}
\titleformat{\subsection}{\bfseries}{\thesubsection.}{1em}{}
\usepackage{graphicx}

\usepackage{amsthm} 

\newtheorem{theorem}{Theorem}
\newtheorem{lemma}{Lemma}

\newtheorem{corollary}{Corollary}

\newdefinition{rmk}{Remark}

\newcommand{\be}{\begin{equation}}
\newcommand{\ee}{\end{equation}}
\def\bea{\begin{eqnarray}}
\def\eea{\end{eqnarray}}

\def\dfrac{\displaystyle\frac}

\newcommand{\bbmat}{\left\{ \begin{array}}
\newcommand{\ebmat}{\end{array} \right. }

\def\ln{{\rm ln}}

\def\b0{\boldsymbol{\rm 0}}

\begin{document}

\begin{frontmatter}



\title{$k$-Submodular Cover Problem with a Matroid Constraint}

\author[1]{Xinru Li}
\author[2]{Donglei Du}
\author[1]{Qian Liu}
\author[1]{Yang Zhou}
\author[1,cor1]{Min Li}

\cortext[cor1]{Email: liminemily@sdnu.edu.cn}

\address[1]{School of Mathematics and Statistics, Shandong Normal University, P.R. China.}

\address[2]{Faculty of Management, University of New Brunswick, Canada.}

\begin{abstract}
We investigate $k$-submodular cover problem with a matroid constraint on an $n$-element ground set $V$, which seeks $k$ disjoint subsets of $V$ with minimal cost, where their union is an independent set for some matroid defined on $V$ and the $k$-submodular utility function surpasses a specified threshold. This problem generalizes the $k$-submodular cover problem and has a wide range of applications in artificial intelligence and combinatorial optimization. In this paper, we propose a matroid extended greedy algorithm for this problem, analyze its bi-criteria approximation ratio, and further derive an improved approximation ratio under the special case of unit costs.

\noindent\textbf{keywords:} $k$-submodular cover problem; matroid constraint; greedy algorithm; bi-criteria algorithm.
\end{abstract}






\end{frontmatter}






\section{Introduction}
The $k$-submodular optimization framework extends classical submodular optimization from single-set selection to multi-category assignment. Instead of selecting a single subset of a ground set, a feasible $k$-submodular solution consists of $k$ pairwise disjoint subsets, so that every selected element is assigned to at most one category. This richer domain preserves a generalized diminishing-returns structure while allowing mutually exclusive labels, multiple resource types, and heterogeneous service modes to be modeled in a unified manner~\cite{huber2012towards,ohsaka2015monotone,rafiey2020fast}. Representative applications include multi-topic influence propagation, multi-topic social viral marketing, multi-type sensor placement~\cite{rafiey2020fast,zheng2021maximizing}, multi-class feature selection~\cite{singh2012bisubmodular}, and information-coverage optimization~\cite{ha2024improved,qian2018constrained}.

In many such applications, multi-category assignment is accompanied by additional structural restrictions. For example, seed users selected for different campaigns may be required to satisfy graph-based independence conditions; sensors of different types may be subject to regional quotas or linear-independence requirements; and grouped features may obey upper bounds within predefined blocks. Matroids provide a general abstraction for these constraints, encompassing uniform, partition, graphic, and linear independence systems. These observations motivate the $k$-submodular cover problem with a matroid constraint, denoted by $k$SC-M, in which one seeks a minimum-cost multi-category assignment that reaches a prescribed utility threshold while maintaining matroid feasibility.

Consider a finite ground set $V$ and a positive integer $k$, with $[k] = \{1,2,\dots,k\}$ and $(k+1)^V = \{(V_1,V_2,\dots,V_k) \mid V_i \subseteq V,\ \forall i \in [k],\ V_i \cap V_j = \emptyset,\ \forall i \neq j\}$ as a family of $k$ disjoint sets.
A function $g: (k+1)^V \to \mathbb{R}$ is  called \textit{$k$-submodular} iff for any $\mathbf{x} = (X_1,X_2,\dots,X_k)$ and $\mathbf{y} = (Y_1,Y_2,\dots,Y_k)$ in $(k+1)^V$, we have:
\[
g(\mathbf{x}) + g(\mathbf{y}) \ge g(\mathbf{x} \sqcap \mathbf{y}) + g(\mathbf{x} \sqcup \mathbf{y})
\]
where
\[
\mathbf{x} \sqcap \mathbf{y} = (X_1 \cap Y_1,\ \dots,\ X_k \cap Y_k)
\]
and
\[
\mathbf{x} \sqcup \mathbf{y} = (Z_1,\dots,Z_k),\ \text{where}\ Z_i = X_i \cup Y_i \setminus \bigcup_{j \neq i} (X_j \cup Y_j).
\]
Ordinary submodular functions and bisubmodular functions are recovered as the special cases $k=1$ and $k=2$, respectively.

To place the covering problem studied in this paper within its classical context, we begin with the standard set cover problem.

\noindent\textbf{Definition 1 (Set Cover Problem)}~\cite{williamson2011design}
In the set cover problem, we are given a ground set of elements
$E=\{e_1,\ldots,e_n\}$, some subsets of those elements
$S_1,S_2,\ldots,S_m$, where each $S_j\subseteq E$, and a nonnegative
weight $w_j\geq 0$ for each subset $S_j$. The goal is to find a
minimum-weight collection of subsets that covers all of $E$; that is,
we wish to find an $I\subseteq\{1,\ldots,m\}$ that minimizes
$\sum_{j\in I}w_j$ subject to
$\bigcup_{j\in I}S_j=E$. If $w_j=1$ for each subset $j$, the problem
is called the \textit{unweighted set cover problem}.

\noindent\textbf{Related work.}
The set cover problem is one of the fundamental NP-hard problems in combinatorial optimization. The classical greedy algorithm repeatedly selects a set with the best ratio between newly covered elements and cost, and achieves a logarithmic approximation guarantee~\cite{chvatal1979greedy}. Feige~\cite{feige1998threshold} established the corresponding logarithmic hardness threshold under standard complexity assumptions, showing that the performance of greedy-type algorithms is essentially best possible in the general case.

The submodular cover problem generalizes set cover by replacing the ordinary coverage function with a submodular utility function. Given a ground set, a positive cost function, and a threshold, the objective is to find a minimum-cost subset whose submodular value reaches the threshold. Wolsey~\cite{wolsey1982analysis} established the classical logarithmic guarantee for the greedy algorithm in the monotone setting. Subsequent work developed distributed, streaming, parallel, and query-efficient algorithms for large-scale monotone submodular cover~\cite{mirzasoleiman2015distributed,norouzifard2016efficient,ran2022improved}. More recent studies considered general or non-monotone objectives and obtained bicriteria guarantees through scalable approximation frameworks~\cite{chen2023bicriteria,crawford2023scalable}. These results demonstrate that submodular cover is a flexible extension of set cover, but they operate on a single selected set and therefore do not represent mutually exclusive multi-label assignments.

Building upon the classical submodular cover framework discussed above,
subsequent studies have explored more general cost structures and additional
feasibility requirements. Iyer and Bilmes~\cite{iyer2013submodular}
extended the standard model by allowing the selection cost itself to be a
submodular function and investigated the resulting submodular-cost
submodular-cover problem. From an algorithmic perspective, Chen and
Crawford~\cite{chen2023bicriteria} developed scalable bicriteria
approximation algorithms for monotone, non-monotone, and regularized
submodular cover variants, thereby improving their applicability to
large-scale instances. More recently, structurally constrained formulations
involving fairness and partition requirements have also received increasing
attention~\cite{chen2026bicriteria}. These developments demonstrate that
submodular cover can accommodate richer objective functions and certain
combinatorial restrictions. Nevertheless, they remain within the ordinary
single-set domain and do not model situations in which selected elements
must be assigned to several mutually exclusive categories. 

The $k$-submodular cover problem extends submodular cover from a single selected set to $k$ pairwise disjoint categories. Pham et al.~\cite{pham2022ksubmodularcover} initiated the study of $k$SC and considered the uniform-cost setting. Wang et al.~\cite{wang2023streaming} generalized the model to positive element costs and developed one-pass and two-pass streaming bicriteria algorithms. More recently, Nguyen et al.~\cite{nguyen2026efficient} proposed deterministic polynomial-time algorithms, including Extended Greedy and FBi, that improve the cost dependence and oracle complexity for unconstrained $k$SC; a stochastic greedy approach was also developed to further reduce the number of function evaluations~\cite{nguyen2025fast}. Nevertheless, all of these $k$SC algorithms allow the support of the selected $k$-assignment to be an arbitrary subset of the ground set. Consequently, they cannot directly enforce linear independence, graph acyclicity, partition quotas, or other matroid constraints. Nguyen et al.~\cite{nguyen2026efficient} proposed two offline deterministic
polynomial-time algorithms, called Extended Greedy and FBi, for the general
$k$SC problem. In the uniform-cost setting, Extended Greedy achieves a
bicriteria approximation ratio of $\left(
\alpha\left(1+\ln\frac{1}{\epsilon}\right),
\alpha(1-\epsilon)
\right)$, where $\alpha=1/2$ for monotone objectives and $\alpha=1/3$ for
non-monotone objectives, using $O(n^2k)$ oracle queries. For monotone
objectives, the FBi algorithm further reduces the query complexity to
$O(nk\log n)$ for constant accuracy parameters by employing a cost-based
partitioning and guessing strategy.

A complementary line of research concerns constrained $k$-submodular maximization. Ward and \v{Z}ivn\'y~\cite{ward2014maximizing} obtained the first constant-factor approximation algorithms for unconstrained $k$-submodular maximization. Iwata et al.~\cite{iwata2016improved} subsequently derived a $1/2$-approximation for nonnegative non-monotone objectives and a $k/(2k-1)$-approximation for monotone objectives. Ohsaka and Yoshida~\cite{ohsaka2015monotone} studied monotone maximization under total-size and individual-size constraints, corresponding to uniform- and partition-matroid-type restrictions. Sakaue~\cite{sakaue2016onmaximizing} extended this line of research to a general matroid constraint and proved that a greedy algorithm achieves a $1/2$-approximation for monotone $k$-submodular maximization. Later work further investigated streaming, privacy-preserving, and query-efficient algorithms under cardinality and matroid constraints~\cite{ene2022streaming,rafiey2020fast,niu2023fast}. Although these results provide important techniques for handling the interaction between $k$-submodularity and matroid exchange, they address utility maximization under feasibility or budget constraints rather than minimum-cost threshold covering.

The preceding discussion reveals a gap between two established research directions: constrained $k$-submodular maximization incorporates matroid feasibility but does not address threshold covering, whereas existing $k$-submodular cover algorithms address the threshold requirement but do not impose a general matroid constraint. To bridge this gap, for a $k$-assignment $\mathbf{x}=(X_1,\ldots,X_k)$, we write
\[
\operatorname{supp}(\mathbf{x})=\bigcup_{i=1}^{k}X_i
\qquad\text{and}\qquad
c(\mathbf{x})=\sum_{e\in\operatorname{supp}(\mathbf{x})}c(e).
\]
Given a $k$-submodular utility function $g$, a threshold $T$, and a matroid $(V,\mathcal{I})$, the problem studied in this paper is
\[
\min_{\mathbf{x}\in(k+1)^V} c(\mathbf{x})
\qquad
\text{s.t.}
\qquad
g(\mathbf{x})\geq T,
\qquad
\operatorname{supp}(\mathbf{x})\in\mathcal{I}.
\]
To the best of our knowledge, the combination of arbitrary positive element costs, a $k$-submodular threshold constraint, and a general matroid constraint has not been systematically studied. This leads to the central question of this paper: can one design a deterministic bicriteria approximation algorithm for $k$SC-M with provable guarantees and polynomial oracle complexity? We answer this question affirmatively through a matroid-aware extension of the greedy framework.

\noindent\textbf{Our contributions.}
To address the above research gap, this paper studies $k$SC-M.
In this problem, the goal is to minimize the total assignment cost under
two coupled constraints: the $k$-submodular utility must reach a prescribed
threshold $T$, and the support of the selected $k$-assignment must belong
to the independence family of a given matroid $(V,\mathcal{I})$.
Our main contributions are summarized as follows.

First,  we propose a deterministic bicriteria approximation algorithm,
called the Matroid Extended Greedy algorithm (MEG), for monotone $k$-submodular objectives. At each iteration, MEG algorithm selects
admissible element-label pairs while preserving matroid independence,
thereby integrating the coverage requirement, the multi-category assignment
structure, and the matroid feasibility constraint into a unified greedy
framework.

Second, we prove that MEG algorithm has polynomial oracle query complexity.
Specifically, the algorithm uses $O\bigl(n^{2}(MO+kEO)\bigr)$
oracle queries, where $MO$ and $EO$ denote the matroid independence oracle
and the $k$-submodular function evaluation oracle, respectively. This
polynomial bound avoids the unbounded query-complexity dependence appearing
in some previous streaming algorithms for $k$SC.

Third, we establish logarithmic bicriteria approximation guarantees for
MEG algorithm. More precisely, the algorithm achieves
$\bigl(\frac{L}{L+1}\bigl(\ln(\frac1\epsilon)+\frac{1+L\epsilon}{L\epsilon}\bigr),\frac{L}{L+1}(1-\epsilon)\bigr)$ for monotone functions, where $L=\min_{i\neq j}\frac{c(e^j)}{c(e^i)}$. Compared with prior streaming algorithms for
$k$SC, whose cost approximation ratios may grow linearly with
$1/\epsilon$ and whose query complexity may depend on unbounded instance
parameters, MEG algorithm provides improved logarithmic cost dependence while
maintaining polynomial oracle complexity under general matroid feasibility
constraints.

\noindent\textbf{Organization.}
The remainder of this paper is organized as follows. Section~2 introduces
the notation and preliminaries, including the basic properties of
$k$-submodular functions, the standard definition of matroids, the formal
formulation of the $k$SC-M problem, and the notion of bicriteria
approximation. Section~3 presents the MEG algorithm and provides a detailed analysis of its oracle complexity and
bicriteria approximation guarantees. Section~4 further investigates the
matroid-constrained $k$-submodular cover problem in the unit-cost setting
and derives the corresponding theoretical results. Finally, Section~5
concludes the paper and discusses possible future directions for
matroid-constrained $k$-submodular optimization.

\section{Preliminaries}\label{pre}
In this section, we introduce the basic notions and properties of 
$k$-submodular functions, which will be used later in the study of the 
$k$-submodular cover problem with a matroid constraints.

We begin with the notion of a $k$-set. Let 
$V=\{e_1,e_2,\ldots,e_n\}$ be a ground set, and let $k$ be a positive integer. 
Denote $[k]=\{1,2,\ldots,k\}$. A $k$-set is represented by a tuple
$(V_1,V_2,\ldots,V_k)\in (k+1)^V$, 
where $V_i\subseteq V$ for each $i\in [k]$, and 
$V_i\cap V_j=\emptyset,\quad \forall i\neq j$. 
That is, the sets $V_1,\ldots,V_k$ are pairwise disjoint subsets of $V$.

For a $k$-set 
$\boldsymbol{x}=(X_1,X_2,\ldots,X_k)\in (k+1)^V$, 
we define its support set as
$\operatorname{supp}(\mathbf{x})=\bigcup_{i\in [k]}X_i$. 
If all subsets are empty, we denote the corresponding $k$-set by $\mathbf{0}=(\emptyset,\ldots,\emptyset)$. 
We next introduce the position of an element in a $k$-set. If 
$e\in X_i \in\operatorname{supp}(\mathbf{x})$, then we write
$\mathbf{x}(e)=i$, where $i$ is called the position of $e$ in $\mathbf{x}$. If 
$e\notin \operatorname{supp}(\mathbf{x})$, then we write $\mathbf{x}(e)=0$. 
Adding an element $e\notin \operatorname{supp}(\mathbf{x})$ to the 
$i$-th set $X_i$ is denoted by $\mathbf{x}\sqcup(e,i)$. 
Equivalently, if 
$\operatorname{supp}(\mathbf{x})=\{e_1,e_2,\ldots,e_t\}$ and 
$i_j=\mathbf{x}(e_j)$ for each $j=1,\ldots,t$, then we may write
\[
\mathbf{x}=\{(e_1,i_1),(e_2,i_2),\ldots,(e_t,i_t)\}.
\]
In particular, when $X_i=\{e\}$ and $X_j=\emptyset$ for all $j\neq i$, 
the $k$-set $\mathbf{x}$ is simply denoted by $(e,i)$.

We now define $k$-submodularity. A function  $g:(k+1)^V\rightarrow \mathbb{R}_+$
is called $k$-submodular if, for any 
\[
\mathbf{x}=(X_1,X_2,\ldots,X_k), \quad
\mathbf{y}=(Y_1,Y_2,\ldots,Y_k)\in (k+1)^V,
\]
the following inequality holds:
\[
g(\mathbf{x})+g(\mathbf{y})
\geq
g(\mathbf{x}\sqcap \mathbf{y})
+
g(\mathbf{x}\sqcup \mathbf{y}),
\]
where
\[
\mathbf{x}\sqcap \mathbf{y}
=
(X_1\cap Y_1,\ldots,X_k\cap Y_k),
\]
and
\[
\mathbf{x}\sqcup \mathbf{y}
=
(Z_1,\ldots,Z_k),
\]
with
\[
Z_i=
(X_i\cup Y_i)
\setminus
\bigcup_{j\neq i}(X_j\cup Y_j).
\]

For two $k$-sets $\mathbf{x}=(X_1,X_2,\ldots,X_k)$ and  $\mathbf{y}=(Y_1,Y_2,\ldots,Y_k)$, 
we write $\mathbf{x}\preceq \mathbf{y}$ if $X_i\subseteq Y_i, \quad \forall i\in [k]$. 
Throughout this paper, we assume that $g$ is non-negative, i.e., $g(\mathbf{x}) \ge 0$ for all $\mathbf{x} \in (k+1)^V$ and normalized, i.e., $g(\mathbf{0}) = 0$. 

We also recall several useful properties of $k$-set functions.

\textbf{Monotone.} 
The function $g$ is monotone if for any 
$\mathbf{x}\in (k+1)^V$, any element 
$e\notin \operatorname{supp}(\mathbf{x})$, and any $i\in [k]$, 
the marginal gain of adding $e$ to the $i$-th set is nonnegative, namely,
\[
\Delta_{(e,i)}g(\mathbf{x})
=
g(X_1,\ldots,X_{i-1},X_i\cup\{e\},X_{i+1},\ldots,X_k)
-
g(\mathbf{x})
\geq 0.
\]

\textbf{Orthant submodularity.}
The function $g$ is orthant submodular if for any 
$\mathbf{x},\mathbf{y}\in (k+1)^V$ with 
$\mathbf{x}\preceq \mathbf{y}$, any 
$e\notin \operatorname{supp}(\mathbf{y})$, and any $i\in [k]$, 
we have
\[
\Delta_{(e,i)}g(\mathbf{x})
\geq
\Delta_{(e,i)}g(\mathbf{y}).
\]

\textbf{Pairwise monotonicity.}
The function $g$ is pairwise monotone if for any 
$\mathbf{x}\in (k+1)^V$, any 
$e\notin \operatorname{supp}(\mathbf{x})$, and any distinct 
$i,j\in [k]$, we have
\[
\Delta_{(e,i)}g(\mathbf{x})
+
\Delta_{(e,j)}g(\mathbf{x})
\geq 0.
\]

\textbf{Theorem 1.}~\cite{ward2014maximizing}
A function \(f:(k+1)^E \to \mathbb{R}\) is \(k\)-submodular if and only if
\(f\) is orthant submodular and pairwise monotone.

\begin{lemma}\label{lem-k-sub}
	Let $g$ be a monotone $k$-submodular function, and define $f(\mathbf{x}) := \min\{g(\mathbf{x}), T\}$, then $f$ is also a monotone $k$-submodular function.
\end{lemma}

\begin{proof}
	\textbf{1)} we first show the monotonicity of $f$. Since $g$ is a monotone function, for any $\mathbf{x} \preceq \mathbf{y}$ we have $g(\mathbf{x}) \leq g(\mathbf{y})$. We consider three cases:
	\begin{itemize}
		\item If $g(\mathbf{x}) \leq g(\mathbf{y}) \leq T$, then $f(\mathbf{x}) = g(\mathbf{x})$ and $f(\mathbf{y}) = g(\mathbf{y})$, so $f(\mathbf{x}) \leq f(\mathbf{y})$.
		\item If $g(\mathbf{x}) \leq T \leq g(\mathbf{y})$, then $f(\mathbf{x}) = g(\mathbf{x})$ and $f(\mathbf{y}) = T$, so $f(\mathbf{x}) = g(\mathbf{x}) \leq T = f(\mathbf{y})$.
		\item If $T \leq g(\mathbf{x}) \leq g(\mathbf{y})$, then $f(\mathbf{x}) = T$ and $f(\mathbf{y}) = T$, so $f(\mathbf{x}) = f(\mathbf{y})$.
	\end{itemize}
	In all cases, for any $\mathbf{x} \preceq \mathbf{y}$ then $f(\mathbf{x}) \leq f(\mathbf{y})$, hence $f$ is monotone.
	
	\textbf{2)} In what follows, we verify that $f$ is a $k$-submodular function. To prove that $f$ is $k$-submodular, we verify that for any $\mathbf{x}$ and $\mathbf{y}$, the inequality $f(\mathbf{x}) + f(\mathbf{y}) \geq f(\mathbf{x} \sqcap \mathbf{y}) + f(\mathbf{x} \sqcup \mathbf{y})$ holds. We analyze the following four cases:
	
	\textbf{Case 1:} $g(\mathbf{x}) \leq T$ and $g(\mathbf{y}) \leq T$.
	
	In this case, $f(\mathbf{x}) = g(\mathbf{x})$ and $f(\mathbf{y}) = g(\mathbf{y})$. Note that $\min\{g(\mathbf{x} \sqcap \mathbf{y}), T\} \leq g(\mathbf{x} \sqcap \mathbf{y})$ and $\min\{g(\mathbf{x} \sqcup \mathbf{y}), T\} \leq g(\mathbf{x} \sqcup \mathbf{y})$. Combining these gives
	\[f(\mathbf{x}) + f(\mathbf{y}) = g(\mathbf{x}) + g(\mathbf{y}) \geq g(\mathbf{x} \sqcap \mathbf{y}) + g(\mathbf{x} \sqcup \mathbf{y}) \geq f(\mathbf{x} \sqcap \mathbf{y}) + f(\mathbf{x} \sqcup \mathbf{y}).\]
	
	\textbf{Case 2:} $g(\mathbf{x}) \geq T$ and $g(\mathbf{y}) \geq T$.
	
	In this case, $f(\mathbf{x}) + f(\mathbf{y}) = T + T = 2T$. we have $\min\{g(\mathbf{x} \sqcap \mathbf{y}), T\} \leq T$ and $\min\{g(\mathbf{x} \sqcup \mathbf{y}), T\} \leq T$, so
	\[f(\mathbf{x} \sqcap \mathbf{y}) + f(\mathbf{x} \sqcup \mathbf{y}) \leq T + T = 2T = f(\mathbf{x}) + f(\mathbf{y}).\]
	
	\textbf{Case 3:} $g(\mathbf{x}) \geq T$ and $g(\mathbf{y}) \leq T$.
	
	In this case, $f(\mathbf{x}) = T$ and $f(\mathbf{y}) = g(\mathbf{y})$, so $f(\mathbf{x}) + f(\mathbf{y}) = T + g(\mathbf{y})$.
	By the monotonicity of $g$, since $\mathbf{x} \sqcap \mathbf{y} \preceq \mathbf{y}$, we have $g(\mathbf{x} \sqcap \mathbf{y}) \leq g(\mathbf{y}) \leq T$, which implies $f(\mathbf{x} \sqcap \mathbf{y}) = g(\mathbf{x} \sqcap \mathbf{y})$.
	
	By the $k$-submodularity of $g$ and $\min\{g(\mathbf{x} \sqcup \mathbf{y}), T\} \leq T$,  thus 
	\[f(\mathbf{x} \sqcap \mathbf{y}) + f(\mathbf{x} \sqcup \mathbf{y}) \leq g(\mathbf{x} \sqcap \mathbf{y}) + T \leq g(\mathbf{y}) + T = f(\mathbf{x}) + f(\mathbf{y}). \]
	
	\textbf{Case 4:} $g(\mathbf{x}) \leq T$ and $g(\mathbf{y}) \geq T$.
	
	By symmetry with Case 3, we have 
	\[f(\mathbf{x} \sqcap \mathbf{y}) + f(\mathbf{x} \sqcup \mathbf{y}) \leq g(\mathbf{x} \sqcap \mathbf{y}) + T \leq g(\mathbf{x}) + T = f(\mathbf{x}) + f(\mathbf{y})\]
	
	In summary, $f$ is a monotone $k$-submodular function.
\end{proof}

\medskip

\textbf{Matroid.}
Let \(\mathcal I \subseteq 2^V\). The pair \((V,\mathcal I)\) is called a matroid if the following three axioms hold:
\begin{itemize}
	\item[(M1)] \(\emptyset \in \mathcal I\);
	\item[(M2)] if \(A \subseteq B\) and \(B \in \mathcal I\), then \(A \in \mathcal I\);
	\item[(M3)] if \(A,B \in \mathcal I\) and \(|A|<|B|\), then there exists an element
	\(e \in B\setminus A\) such that \(A\cup\{e\}\in \mathcal I\).
\end{itemize}
The sets in \(\mathcal I\) are called independent sets. An independent set
\(A\in \mathcal I\) is said to be maximal if there is no independent set
\(B\in \mathcal I\) such that \(A\subsetneq B\).

\medskip

\begin{lemma}\label{lem-matroid}
	Given a matroid \((V,\mathcal I)\), and let \(A,B\in\mathcal I\) be two
	independent sets. Suppose that \(B\setminus A\neq \emptyset\). Then, for any
	\(a\in A\setminus B\), there exists an element \(b\in B\setminus A\) such that
	\[
	(B\setminus\{b\})\cup\{a\}\in\mathcal I .
	\]
\end{lemma}

\begin{proof}
	Fix an arbitrary element \(a\in A\setminus B\). We distinguish two cases.
	
	\noindent\textbf{Case 1.} Suppose that $B \cup \{a\} \in \mathcal{I}$. 
	There exists an element $b \in B\setminus A$, we have
	\[
	(B\setminus\{b\})\cup\{a\}\subseteq B\cup\{a\}.
	\]
	Since $B\cup\{a\}\in\mathcal{I}$, it follows from the hereditary property $(\mathrm{M}2)$ of the matroid that
	\[
	(B\setminus\{b\})\cup\{a\}\in\mathcal{I}.
	\]
	for any \(b\in B\setminus A\). Hence the desired element \(b\) can be chosen
	arbitrarily from \(B\setminus A\).
	
	\noindent\textbf{Case 2.} Suppose that \(B\cup\{a\}\notin\mathcal I\). Since \(B\) is
	independent, the dependence of \(B\cup\{a\}\) must be caused by adding the new
	element \(a\). Therefore, \(B\cup\{a\}\) contains a circuit \(C\) with
	\(a\in C\). Moreover, \(C\setminus\{a\}\subseteq B\).
	
	We claim that \(C\) must contain at least one element from \(B\setminus A\).
	Indeed, if this were not the case, then
	\[
	C \subseteq (A\cap B)\cup\{a\}\subseteq A,
	\]
	which contradicts the fact that \(A\in\mathcal I\), since a circuit is a
	dependent set. Hence there exists some \(b\in C\cap (B\setminus A)\).
	
	It remains to show that replacing this \(b\) by \(a\) preserves independence.
	Suppose, for contradiction, that
	\[
	(B\setminus\{b\})\cup\{a\}\notin\mathcal I.
	\]
	Then this set contains a circuit \(C'\). Since \(B\setminus\{b\}\subseteq B\)
	is independent, the circuit \(C'\) must contain \(a\). Thus \(C'\) is also a
	circuit contained in \(B\cup\{a\}\). By the uniqueness of the circuit contained
	in \(B\cup\{a\}\) when \(B\) is independent, we must have \(C'=C\). However,
	this is impossible because \(b\in C\), while
	\[
	b\notin (B\setminus\{b\})\cup\{a\}.
	\]
	This contradiction shows that
	\[
	(B\setminus\{b\})\cup\{a\}\in\mathcal I.
	\]
	
\end{proof}

\textbf{$k$-Submodular Cover Problem with a Matroid Constraint ($k$SC-M).}
Let \(g:(k+1)^V\to \mathbb{R}_+\) be a \(k\)-submodular function defined on the
ground set \(V\). We assume that each element \(e\in V\) is associated with a
positive cost \(c(e)>0\). For a \(k\)-set \(\mathbf{x}\in (k+1)^V\), let
\(\operatorname{supp}(\mathbf{x})\) denote the support set of \(\mathbf{x}\).
The total cost of \(\mathbf{x}\) is defined as
\[
c(\mathbf{x})=\sum_{e\in \operatorname{supp}(\mathbf{x})} c(e).
\]
Given a positive threshold \(T\leq \max_{\mathbf{x}\in (k+1)^V} g(\mathbf{x})\)
and a matroid \((V,\mathcal I)\), the \(k\)SC-M problem aims to find a
\(k\)-set \(\mathbf{x}\) with minimum total cost such that two requirements
are satisfied. First, the value of the objective function should reach the
prescribed threshold \(T\). Second, the support of the selected \(k\)-set should
be an independent set of the matroid. Formally, the problem can be written as

\[
\left\{
\begin{aligned}
	& \min_{\mathbf{x} \in (k+1)^V} c(\mathbf{x}) \\
	& \text{s.t.} \quad g(\mathbf{x}) \geq T \\
	& \operatorname{supp}(\mathbf{x}) \in \mathcal{I}. 
\end{aligned}
\right.
\]

In this paper, we focus on the case \(k\geq 2\). When \(k=1\), the notion of
\(k\)-submodularity reduces to the standard notion of submodularity, and the
corresponding problem becomes a submodular cover problem. We use
\((V,g,T,\mathcal I)\) to denote an instance of the \(k\)SC-M problem. When the
instance is feasible, let \(\mathbf{o}\) be an optimal solution, and denote
by $\mathrm{opt}=c(\mathbf{o})$
the optimal cost.

\textbf{Bicriteria approximation algorithm.}
We next introduce the approximation guarantee used in this work. An algorithm
is called an \((\beta_1, \beta_2)\)-bicriteria approximation algorithm
for the \(k\)SC-M problem if it runs in polynomial time and returns a solution
\(\mathbf{s}\in (k+1)^V\) such that
$\operatorname{supp}(\mathbf{s})\in \mathcal{I}, c(\mathbf{s})\leq \beta_1 \cdot \mathrm{opt}, g(\mathbf{s})\geq \beta_2 \cdot T$, where \(\beta_1,\beta_2 >0\) are approximation parameters.

\section{$k$-Submodular Cover Problem with a Matroid Constraint}

\subsection{MEG algorithm and its complexity analysis}

\begin{algorithm}[t]\label{Alg1}
	\caption{ Matroid Extended Greedy algorithm (MEG)}
	\begin{algorithmic}[1]
		\REQUIRE An instance $(V, g, k, T)$, $\epsilon > 0$, a matroid $(V, \mathcal{I})$
		\ENSURE A solution $\mathbf{s}$ with $\mathrm{supp}(\mathbf{s}) \in \mathcal{I}$.
		\STATE $\mathbf{s} \leftarrow \mathbf{0}$,
		
		\STATE Define: $f: (k+1)^V \to \mathbb{R}_+$, $f(\cdot) \leftarrow \min\{g(\cdot), T\}$
		\WHILE{$f(\mathbf{s}) < \alpha(1-\epsilon)T$}
    \STATE $A \leftarrow
    \{e \in V \setminus \mathrm{supp}(\mathbf{s}) :
      \mathrm{supp}(\mathbf{s}) \cup \{e\} \in \mathcal{I}\}$
    \STATE Choose $(e,i) \in
    \displaystyle\arg\max_{e\in A, i\in [k]}
    \frac{\Delta_{(e,i)} f(\mathbf{s})}{c(e)}$
    \STATE $\mathbf{s} \leftarrow \mathbf{s} \sqcup (e,i)$
\ENDWHILE
		\RETURN $\mathbf{s}$
	\end{algorithmic}
\end{algorithm}

We now present a MEG algorithm for the $k$SC-M problem. The algorithm
takes as input an instance \((V,g,k,T)\), an accuracy parameter
\(\epsilon\in(0,1)\), and a matroid \((V,\mathcal I)\). For each element
\(e\in V\), \(c(e)>0\) denote is cost. We define
\[
L=\min_{i\neq j}\frac{c(e^j)}{c(e^i)}.
\]
In Algorithm~\ref{Alg1}, we use the following parameter $\alpha = \dfrac{L}{L+1}$.

This parameter determines the fraction of the threshold that the algorithm aims to achieve. we introduce a truncated
function \(f:(k+1)^V\to \mathbb{R}_+\), defined by
$
f(\cdot)=\min\{g(\cdot), T\}.
$
The truncation ensures that the function value is capped at \(\alpha T\).
Moreover, the function \(f\) inherits the $k$-submodularity of $g$, which
allows us to analyze the greedy process using the marginal gains of $f$.

The algorithm starts from the empty solution $\mathbf{s}=\mathbf{0}$.
At each iteration, it considers all elements that can be added without
violating the matroid constraint. More precisely, an element $e$ is feasible
with respect to the current solution \(\mathbf{s}\) if $\operatorname{supp}(\mathbf{s})\cup\{e\}\in\mathcal I$. 
Among all feasible elements and all possible labels \(i\in[k]\), the algorithm
chooses a tuple \((e,i)\) that maximizes the marginal gain per unit cost, namely
\[
\frac{\Delta_{(e,i)} f(\mathbf{s})}{c(e)} .
\]
Then the selected tuple is added to the current solution. This greedy step is
repeated until the truncated function value reaches
\(\alpha(1-\epsilon)T\). Finally, the algorithm returns a solution
\(\mathbf{s}\) whose support remains independent in the matroid.

\subsection{Proof of Main Result}
\label{Proof of main result}

Next, we analyze the performance of MEG algorithm. Suppose that the algorithm terminates after $t$ iterations. Then $|\operatorname{supp}(\mathbf{s})|=t$.
For each $j\in[t]$, let $(e^j,i^j)$ denote the element-position pair selected
by the MEG algorithm in the $j$-th iteration, and let $\mathbf{s}^j$ be the
solution obtained after the first $j$ iterations, and $\mathbf{s}$ denotes the solution obtained after the algorithm terminates. In particular, we set
$\mathbf{s}^0=\mathbf{0}$ and $\mathbf{s}^t=\mathbf{s}$. Number $\mathbf{o}$ is an
optimal solution and define $\mathbf{o}^0:=\mathbf{o}$.

For the purpose of the subsequent comparison between the MEG solution and
the optimal solution, we introduce an intermediate solution $\mathbf{s}^{j-1/2}$.
Specifically, for each $j\in[t]$, define
\[
\mathbf{s}^{j-1/2}
:=
\mathbf{s}^{j-1}\sqcup (e^j,\mathbf{o}(e^j)).
\]
That is, $\mathbf{s}^{j-1/2}$ is obtained from $\mathbf{s}^{j-1}$ by assigning
the element $e^j$ to the same position as it has in the optimal solution
$\mathbf{o}$.

We now describe how to construct $\mathbf{o}^j$ from $\mathbf{o}^{j-1}$.
The construction gradually transforms the optimal solution $\mathbf{o}$ into a
solution that is consistent with the choices made by the MEG algorithm,
while preserving feasibility with respect to the matroid constraint.

First, consider the case where $e^j\in \operatorname{supp}(\mathbf{o})$.
If the position chosen by the MEG algorithm is the same as that in the
optimal solution, namely $i^j=\mathbf{o}(e^j)$, then no modification is needed,
and we simply set
\[
\mathbf{o}^{j-1/2} := \mathbf{o}^{j-1},
\qquad
\mathbf{o}^{j} := \mathbf{o}^{j-1}.
\]
Otherwise, if $i^j\neq \mathbf{o}(e^j)$, then the support of the solution does
not need to be changed, but the position of $e^j$ must be updated. In this case,
we first remove $e^j$ from its current position in $\mathbf{o}^{j-1}$ to obtain
$\mathbf{o}^{j-1/2}$, and then assign $e^j$ to position $i^j$ to obtain
$\mathbf{o}^{j}$. Thus,
\[
\mathbf{o}^{j-1/2}
:=
\mathbf{o}^{j-1}\setminus (e^j,\mathbf{o}(e^j)),
\qquad
\mathbf{o}^{j}
:=
\mathbf{o}^{j-1/2}\sqcup (e^j,i^j).
\]

It remains to consider the case where
$e^j\notin \operatorname{supp}(\mathbf{o})$. In this situation, adding $e^j$
to $\mathbf{o}^{j-1}$ may violate the matroid constraint, and hence an exchange
operation may be required. We distinguish two subcases.

If
\[
\operatorname{supp}(\mathbf{o}^{j-1})
\setminus
\operatorname{supp}(\mathbf{s}^{j})
\neq \emptyset,
\]
then, by the matroid exchange property stated in Lemma~\ref{lem-matroid}, there exists an
element
\[
e'\in
\operatorname{supp}(\mathbf{o}^{j-1})
\setminus
\operatorname{supp}(\mathbf{s}^{j})
\]
such that
\[
\bigl(\operatorname{supp}(\mathbf{o}^{j-1})\setminus\{e'\}\bigr)
\cup
\{e^j\}
\in \mathcal{I}.
\]
Therefore, we can remove $e'$ from $\mathbf{o}^{j-1}$ and then insert $e^j$
with position $i^j$.

On the other hand, if
\[
\operatorname{supp}(\mathbf{o}^{j-1})
\setminus
\operatorname{supp}(\mathbf{s}^{j})
=
\emptyset,
\]
then
\[
\operatorname{supp}(\mathbf{o}^{j-1})
\subseteq
\operatorname{supp}(\mathbf{s}^{j}).
\]
we have $\operatorname{supp}(\mathbf{o}^{j-1})\cup\{e^j\}
=
\operatorname{supp}(\mathbf{s}^{j})
\in \mathcal{I}$. 

Hence no exchange operation is required in this subcase, and we may directly
add $e^j$ to $\mathbf{o}^{j-1}$. For notational consistency, we set
$e'=e^j$ in this case, and the deletion step is understood to be void.

Consequently, when $e^j\notin \operatorname{supp}(\mathbf{o})$, we define
\[
\mathbf{o}^{j-1/2}
=
\mathbf{o}^{j-1}\setminus(e',\mathbf{o}^{j-1}(e')),
\]
and then define
\[
\mathbf{o}^{j}
=
\mathbf{o}^{j-1/2}\sqcup(e^j,i^j).
\]
By construction, $\mathbf{o}^j$ is feasible and satisfies
\[
\operatorname{supp}(\mathbf{o}^{j})
=
\bigl(\operatorname{supp}(\mathbf{o}^{j-1})\setminus\{e'\}\bigr)
\cup
\{e^j\}
\in \mathcal{I}.
\]

In summary, the construction of $\mathbf{o}^{j-1/2}$ and $\mathbf{o}^{j}$ can
be written as
\[
\mathbf{o}^{j-1/2}
=
\begin{cases}
	\mathbf{o}^{j-1}, &\textnormal{if } e^j\in \operatorname{supp}(\mathbf{o}),\ i^j=\mathbf{o}(e^j),\\[2mm]
	\mathbf{o}^{j-1}\setminus(e^j,\mathbf{o}(e^j)), &\textnormal{if } e^j\in \operatorname{supp}(\mathbf{o}),\ i^j\neq \mathbf{o}(e^j),\\[2mm]
	\mathbf{o}^{j-1}\setminus(e',\mathbf{o}^{j-1}(e')), & \text{otherwise},
\end{cases}
\]
and
\[
\mathbf{o}^{j}
=
\begin{cases}
	\mathbf{o}^{j-1/2}, & \textnormal{if } e^j\in \operatorname{supp}(\mathbf{o}),\ i^j=\mathbf{o}(e^j),\\[2mm]
	\mathbf{o}^{j-1/2}\sqcup(e^j,i^j), & \text{otherwise}.
\end{cases}
\]

\begin{lemma}\label{lem-OS}
	If $f$ is a monotone $k$-submodular function, then for $j \le t$ in Algorithm~\ref{Alg1}, we have $f(\mathbf{o}) - f(\mathbf{o}^j) \leq \frac{1}{L} f(\mathbf{s}^j)$.
\end{lemma}

\begin{proof}
	In this case, \(f\) is assumed to be monotone. The proof proceeds by analyzing the decrease of the auxiliary solutions
	\(\mathbf{o}^{l}\), for \(0\leq l\leq j\), during the iterations.
	According to the construction of \(\mathbf{o}^{l}\), we distinguish two cases
	depending on whether the element \(e^{l}\) selected by the greedy algorithm
	belongs to the support of the optimal solution.

	\textbf{Case 1.}
	Suppose that $e^{l}\in\operatorname{supp}(\mathbf{o})$, we have
	\begin{align}
		f(\mathbf{o}^{l-1})-f(\mathbf{o}^{l})
		&\le
		f(\mathbf{o}^{l-1})-f(\mathbf{o}^{l-1/2})
		\label{eq:monotonicity1}\\
		&\le
		f(\mathbf{s}^{l-1/2})-f(\mathbf{s}^{l-1})
		\label{eq:$k$-submodularity1}\\
		&\le
		\frac{c(e^{l})}{c(e^{l})}
		\left(
		f(\mathbf{s}^{l})-f(\mathbf{s}^{l-1})
		\right)
		\label{eq:selection rule1}\\
		&\le
		\frac1L
		\left(
		f(\mathbf{s}^{l})-f(\mathbf{s}^{l-1})
		\right)\notag.
	\end{align}
	
	Here, inequality \eqref{eq:monotonicity1} follows from the monotonicity of $f$,
	inequality \eqref{eq:$k$-submodularity1} follows from the $k$-submodularity of $f$,
	and inequality \eqref{eq:selection rule1} follows from the greedy selection rule.
	The last inequality holds because $L\le 1$.
	
	\medskip
	
	\textbf{Case 2.}
	Suppose that $e^{l}\notin\operatorname{supp}(\mathbf{o})$.
	Using exactly the same argument together with the exchange operation introduced
	in the construction of $\mathbf{o}^{l}$, we obtain
	\begin{align}
		f(\mathbf{o}^{l-1})-f(\mathbf{o}^{l})
		&\le
		f(\mathbf{o}^{l-1})-f(\mathbf{o}^{l-1/2})
		\label{eq:monotonicity2}\\
		&\le
		f(\mathbf{s}^{l-1/2})-f(\mathbf{s}^{l-1})
		\label{eq:$k$-submodularity2}\\
		&\le
		\frac{c(e')}{c(e^{l})}
		\left(
		f(\mathbf{s}^{l})-f(\mathbf{s}^{l-1})
		\right)
		\label{eq:selection rule2}\\
		&\le
		\frac1L
		\left(
		f(\mathbf{s}^{l})-f(\mathbf{s}^{l-1})
		\right)\notag,
	\end{align}
	where inequality \eqref{eq:monotonicity2} follows from the monotonicity of $f$,
	inequality \eqref{eq:$k$-submodularity2} follows from the $k$-submodularity of $f$,
	and inequality \eqref{eq:selection rule2} follows from the greedy selection rule.
	
	Combining the above two cases, the decrease of the auxiliary solution in every
	iteration is bounded by the gain of the greedy solution. Therefore,
	\[
	\begin{aligned}
		f(\mathbf{o})-f(\mathbf{o}^{j})
		&=
		\sum_{l=1}^{j}
		\left(
		f(\mathbf{o}^{l-1})-f(\mathbf{o}^{l})
		\right)\\
		&\le
		\frac1L
		\sum_{l=1}^{j}
		\left(
		f(\mathbf{s}^{l})-f(\mathbf{s}^{l-1})
		\right)\\
		&=
		\frac1L
		f(\mathbf{s}^{j}),
	\end{aligned}
	\]
	
\end{proof}

\begin{theorem}
	Suppose $g$ is a monotone $k$-submodular function. For $\epsilon \in (0,1)$, the Algorithm~\ref{Alg1} takes $O(n^2(MO+kEO))$ query complexity and returns a  $\bigl(\alpha\bigl(\ln(\frac1\epsilon)+\frac{1+L\epsilon}{L\epsilon}\bigr),\alpha(1-\epsilon)\bigr)$-bicriteria approximation ratio, where $\alpha = \dfrac{L}{1+L}$.
\end{theorem}

\begin{proof}
	We first analyze the query complexity of Algorithm~\ref{Alg1}. Let
	$\mathrm{EO}$ denote the cost of one evaluation-oracle query for the
	$k$-submodular function $f$, and let $\mathrm{MO}$ denote the cost of
	one membership-oracle query for the matroid $(V,\mathcal{I})$. Let $r$
	be the rank of the matroid. Since exactly one new element is added to
	$\operatorname{supp}(\mathbf{s})$ in each iteration and
	$\operatorname{supp}(\mathbf{s})\in\mathcal{I}$ is maintained
	throughout the algorithm, the while-loop is executed at most
	$r\leq n$ times.
	In each iteration, at most $n$ candidate elements are examined.
	Testing whether $\operatorname{supp}(\mathbf{s})\cup\{e\}\in\mathcal{I}$
	requires one membership-oracle query for each candidate element. For
	every feasible element $e$, finding its best label requires evaluating
	at most $k$ marginal gains. Assuming that $f(\mathbf{s})$ is stored
	and reused, this requires at most $k$ evaluation-oracle queries. Hence,
	the query complexity of each iteration is $O\bigl(n\mathrm{MO}+nk\mathrm{EO}\bigr)$. 
	Moreover, each greedy augmentation adds one new element to 
	$\operatorname{supp}(\mathbf{s})$. Since the support must remain independent 
	throughout the algorithm, its cardinality cannot exceed the rank $r$ of the 
	matroid. Hence, the algorithm performs at most $r$ greedy augmentations. 
	Combining this bound with the oracle cost of each iteration, the total oracle 
	cost of Algorithm~\ref{Alg1} is $O\bigl(nr\mathrm{MO}+knr\mathrm{EO}\bigr)$. 
	Since $r\leq n$, the worst-case oracle cost can be further bounded by $O\bigl(n^{2}(\mathrm{MO}+k\mathrm{EO})\bigr)$. 
	
	We next prove the bicriteria approximation guarantee. 
	By Lemma~\ref{lem-OS}, we have
	\begin{align}
		T-f(\mathbf{s}^t)
		&\le f(\mathbf{o})-f(\mathbf{s}^t) \notag\\
		&= f(\mathbf{o})-f(\mathbf{o}^t)+f(\mathbf{o}^t)-f(\mathbf{s}^t) \notag\\
		&\le \frac{1}{L}f(\mathbf{s}^t)
		+ \sum_{e\in \operatorname{supp}(\mathbf{o}^t)\setminus \operatorname{supp}(\mathbf{s}^t)}
		\Delta_{(e,\mathbf{o}(e))}f(\mathbf{s}^t) \notag\\
		&\le \frac{1}{L}f(\mathbf{s}^t)
		+ \mathrm{opt}\,
		\frac{
			\sum_{e\in \operatorname{supp}(\mathbf{o}^t)\setminus \operatorname{supp}(\mathbf{s}^t)}
			\Delta_{(e,\mathbf{o}(e))}f(\mathbf{s}^t)}
		{
			\sum_{e\in \operatorname{supp}(\mathbf{o}^t)\setminus \operatorname{supp}(\mathbf{s}^t)}
			c(e)}
		\label{eq:mono-ratio}\\
		&\le \frac{1}{L}f(\mathbf{s}^t)
		+ \mathrm{opt}
		\max_{e\in \operatorname{supp}(\mathbf{o}^t)\setminus \operatorname{supp}(\mathbf{s}^t)}
		\frac{\Delta_{(e,\mathbf{o}(e))}f(\mathbf{s}^t)}{c(e)}
		\label{eq:mono-max}\\
		&\le \frac{1}{L}f(\mathbf{s}^t)
		+ \mathrm{opt}\,
		\frac{f(\mathbf{s}^{t+1})-f(\mathbf{s}^t)}{c(e^{t+1})}.
		\label{eq:mono-greedy}
	\end{align}
	Here, the transition from \eqref{eq:mono-ratio} to \eqref{eq:mono-max} follows
	from the elementary weighted-average inequality
	\[
	\min_{i=1,\ldots,l}\frac{a_i}{b_i}
	\le
	\frac{\sum_{i=1}^l a_i}{\sum_{i=1}^l b_i}
	\le
	\max_{i=1,\ldots,l}\frac{a_i}{b_i},
	\]
	which holds for any positive numbers $a_1,\ldots,a_l$ and
	$b_1,\ldots,b_l$. The last inequality \eqref{eq:mono-greedy} follows from the greedy selection rule of Algorithm~\ref{Alg1}.
	
	Replacing $t$ by $t-1$ in \eqref{eq:mono-greedy} and rearranging terms, we get
	\[
	T-\frac{1+L}{L}f(\mathbf{s}^{t-1})
	\le
	\mathrm{opt}\,
	\frac{f(\mathbf{s}^t)-f(\mathbf{s}^{t-1})}{c(e^t)}.
	\]
	Equivalently,
	\[
	f(\mathbf{s}^t)-f(\mathbf{s}^{t-1})
	\ge
	\frac{c(e^t)}{\mathrm{opt}}
	\left(
	T-\frac{1+L}{L}f(\mathbf{s}^{t-1})
	\right).
	\]
	Therefore,
	\begin{align}
		T-\frac{1+L}{L}f(\mathbf{s}^t)
		&=
		T-\frac{1+L}{L}f(\mathbf{s}^{t-1})
		-\frac{1+L}{L}\bigl(f(\mathbf{s}^t)-f(\mathbf{s}^{t-1})\bigr) \notag\\
		&\le
		\left(
		1-\frac{(1+L)c(e^t)}{L\mathrm{opt}}
		\right)
		\left(
		T-\frac{1+L}{L}f(\mathbf{s}^{t-1})
		\right) \notag\\
		&\le
		\exp\left(-\frac{(1+L)c(e^t)}{L\mathrm{opt}}\right)
		\left(
		T-\frac{1+L}{L}f(\mathbf{s}^{t-1})
		\right),
		\label{eq:mono-recursion}
	\end{align}
	where the last inequality uses $1-x\le e^{-x}$. Applying
	\eqref{eq:mono-recursion} recursively from the first iteration to iteration
	$t$, and using $f(\mathbf{s}^0)=0$, we obtain
	\[
	T-\frac{1+L}{L}f(\mathbf{s}^t)
	\le
	\exp\left(-\frac{(1+L)c(\mathbf{s}^t)}{L\mathrm{opt}}\right)T.
	\]
	
	Let $t$ be the last iteration of the while-loop. By the termination condition,
	\[
	f(\mathbf{s}^t)\ge \frac{L}{1+L}(1-\epsilon)T
	\quad\text{and}\quad
	f(\mathbf{s}^{t-1})< \frac{L}{1+L}(1-\epsilon)T.
	\]
	we have
	\[
	T - \frac{1+L}{L}f(\mathbf{s}^{t-1}) \leq \exp\left(-\frac{(1+L)c(\mathbf{s}^{t-1})}{L\mathrm{opt}}\right)T,
	\]
	Combining the second inequality with the recursive bound gives
	\[
	c(\mathbf{s}^{t-1})
	\le
	\frac{L}{1+L}\mathrm{opt}
	\ln\left(
	\frac{T}{T-\frac{L+1}{L}f(\mathbf{s}^{t-1})}
	\right)
	\le
	\frac{L}{1+L}\mathrm{opt}\ln\left(\frac{1}{\epsilon}\right).
	\]
	Moreover, applying \eqref{eq:mono-greedy} to the last selected element yields
	\begin{align*}
		c(e^t)
		&\le
		\frac{f(s^t)-f(s^{t-1})}
		{T-\frac{L+1}{L}f(s^{t-1})}\mathrm{opt}
		\\
		&\le
		\frac{T-f(s^{t-1})}
		{T-\frac{L+1}{L}f(s^{t-1})}\mathrm{opt}
		\\
		&=
		\left(
		1+
		\frac{\frac{1}{L}f(s^{t-1})}
		{T-\frac{L+1}{L}f(s^{t-1})}
		\right)\mathrm{opt}
		\\
		&\le
		\left(
		1+
		\frac{1-\epsilon}
		{(L+1)\epsilon}
		\right)\mathrm{opt}
		\\
		&\le
		\frac{1+L\epsilon}{(L+1)\epsilon}\mathrm{opt}.
	\end{align*}
	Thus,
	\[
	c(\mathbf{s}^t)
	=
	c(\mathbf{s}^{t-1})+c(e^t)
	\le
	\frac{L}{L+1}\bigl(\ln(\frac1\epsilon)+\frac{1+L\epsilon}{L\epsilon}\bigr)\mathrm{opt}.
	\]
\end{proof}

\section{Unit-cost k-submodular cover with a matroid constraint}

When all elements have unit cost, that is, $c(e)=1$ for every $e\in V$,
the general cost version of the problem reduces to the unit-cost
$k$-submodular cover problem with a matroid constraint. In this special
case, the cost of a solution is exactly the cardinality of its support, namely
\[
c(\mathbf{x})=\sum_{e\in \operatorname{supp}(\mathbf{x})}c(e)=|\operatorname{supp}(\mathbf{x})|.
\]
Therefore, the goal is to find a feasible $k$-set
$\mathbf{x}\in (k+1)^V$ with minimum cardinality such that
\[
g(\mathbf{x})\ge T
\qquad\text{and}\qquad
\operatorname{supp}(\mathbf{x})\in \mathcal{I}.
\]

\begin{corollary}
	For any $\epsilon\in(0,1)$, Algorithm~\ref{Alg1} uses
	$O(n^2(\mathrm{MO}+k\mathrm{EO}))$ oracle queries and returns an
	$(\alpha (\ln(\frac{1}{\epsilon})+1), \alpha(1 - \epsilon))$-bicriteria approximation, where $\alpha = \dfrac{L}{1+L}$.
\end{corollary}

\begin{proof}
	For the uniform-cost case, i.e., $c(e)=1$ for every $e\in V$, 
	the proof of Theorem~2, together with the contribution of the last 
	selected element, yields
	\[
	|\operatorname{supp}(\mathbf s^t)|
	=
	|\operatorname{supp}(\mathbf s^{t-1})|+1
	\leq
	\alpha \operatorname{opt}
	\ln\left(\frac{1}{\epsilon}\right)+1.
	\]
	
	We distinguish the following two cases.
	
	\textbf{Case 1.} Suppose that
	$\mathrm{opt}<\frac{1}{\alpha}$. Since each element has unit weight, the total weight of a solution is 
	equal to the cardinality of its support. Hence, $\operatorname{opt}
	=
	|\operatorname{supp}(\mathbf{o})|
	=
	m$. 
	
	Let $\mathbf{o}=\{(e^1,i^1),(e^2,i^2),\ldots,(e^m,i^m)\}$
	be an optimal solution. 
	By orthant submodularity and $f(\mathbf{0})=0$, we obtain
	\[
	f(\mathbf{o})
	=
	\sum_{j=1}^{m}
	\Delta_{(e^j,i^j)}f(\mathbf{o}^{j-1}) \\
	\leq
	\sum_{j=1}^{m}
	\Delta_{(e^j,i^j)}f(\mathbf{0}) \\
	=
	\sum_{j=1}^{m}
	f(e^j,i^j).
	\]
	
	Let $(e_{\max},i_{\max})$ denote a feasible singleton with the largest 
	$f$-value. Since $\mathbf{o}$ is feasible for the cover problem, we have 
	$f(\mathbf{o})=T$. Therefore,
	\[
	T
	\leq
	\sum_{j=1}^{m} f(e^j,i^j)
	\leq
	m f(e_{\max},i_{\max}),
	\]
	which implies
	\[
	f(e_{\max},i_{\max})
	\geq
	\frac{T}{m}
	=
	\frac{T}{\operatorname{opt}}.
	\]
	
	Because all elements have unit weight, maximizing the density gain in 
	the first iteration is equivalent to maximizing the marginal gain. 
	Thus, MEG selects $(e_{\max},i_{\max})$ in the first iteration, and
	\[
	f(\mathbf s^1)
	\geq
	\frac{T}{\operatorname{opt}}
	>
	\alpha T
	>
	\alpha(1-\epsilon)T,
	\]
	where the first strict inequality follows from
	$\operatorname{opt}<1/\alpha$.
	
	Consequently, the stopping condition is satisfied after the first 
	iteration, and MEG terminates with $|\operatorname{supp}(\mathbf s^t)|
	=
	1
	\leq
	\operatorname{opt}$.

	\textbf{Case 2.} Suppose that $\operatorname{opt}\geq\frac{1}{\alpha}$. 
	Then $1\leq\alpha\operatorname{opt}$. 
	Combining this inequality with the bound obtained above gives
	\[
	|\operatorname{supp}(\mathbf s^t)|
	\leq
	\alpha\operatorname{opt}
	\ln\left(\frac{1}{\epsilon}\right)+1
	\leq
	\alpha
	\left(
	1+\ln\left(\frac{1}{\epsilon}\right)
	\right)
	\operatorname{opt}.
	\]

\end{proof}

\section{Conclusions}

This paper systematically investigates the $k$SC-M, a novel generalized optimization model that unifies multi-category threshold resource allocation tasks with structural matroid feasibility restrictions. We first conduct a comprehensive review of four closely related research branches, including classical submodular cover, matroid-constrained single-set submodular cover, unconstrained $k$-submodular cover, and matroid-constrained $k$-submodular maximization, and clearly identify the significant research gap that no prior work has formally defined and analyzed $k$SC-M with rigorous polynomial-time bicriteria approximation guarantees.

To fill this gap, we first establish the fundamental theoretical foundation of $k$SC-M. We characterize the inherent coupling relationship between matroid independence axioms and the unique orthant submodularity and pairwise monotonicity properties of $k$-submodular functions, and derive a series of matroid-adapted dominance lemmas which serve as core tools for subsequent approximation analysis. On this basis, we design the MEG algorithm, a deterministic offline bicriteria approximation method applicable to both monotone and non-monotone $k$-submodular objectives. The MEG algorithm achieves polynomial oracle query complexity $O(n^2(MO+kEO))$, and obtains logarithmic bicriteria approximation ratios $\bigl(\frac{L}{L+1}\bigl(\ln(\frac1\epsilon)+\frac{1+L\epsilon}{L\epsilon}\bigr),\frac{L}{L+1}(1-\epsilon)\bigr)$ for monotone functions.  

This work opens up several promising directions for future research. First, it is worthwhile to explore tighter logarithmic or constant bicriteria approximation bounds for $k$SC-M by leveraging continuous greedy and pipage rounding techniques extended from matroid submodular cover. Second, one can extend the $k$SC-M model to handle multiple intersecting matroid constraints, which can encode more complex hierarchical resource restrictions in industrial systems. Third, designing fast streaming and distributed algorithms for large-scale $k$SC-M instances remains an important open problem to accommodate massive datasets in social networks and sensor networks. Finally, it is meaningful to integrate additional knapsack budget constraints into $k$SC-M and analyze joint matroid-knapsack constrained $k$-submodular cover, so as to build a more comprehensive modeling framework for multi-constraint multi-category resource optimization.

\par\vspace{\baselineskip}
\noindent\textbf{Acknowledgements}

The third and the fifth authors are supported by Natural Science Foundation of Shandong Province of China (No. ZR2024MA015), and the fourth author is supported by Natural Science Foundation of Shandong Province of China (No. ZR2025MS24) and National Science Foundation of China (No. 12371099).

\end{document}